\documentclass[aps,pra,twocolumn,showpacs,superscriptaddress]{revtex4-1}

\usepackage{amssymb,amsmath,amsfonts,amsthm}
\usepackage{bbold}
\usepackage{hyperref}
\usepackage{mathrsfs}
\usepackage[english]{babel}
\usepackage{graphicx,bm,tikz,array,latexsym}

\usepackage{listings}
\usepackage{color}

\definecolor{dkgreen}{rgb}{0,0.6,0}
\definecolor{gray}{rgb}{0.5,0.5,0.5}
\definecolor{mauve}{rgb}{0.58,0,0.82}

\newtheorem{theorem}{Theorem}[section]

\newcommand{\abs}[1]{\left| #1 \right|} 
\newcommand{\avg}[1]{\left< #1 \right>} 
\newcommand{\ket}[1]{\left| #1 \right>} 
\newcommand{\bra}[1]{\left< #1 \right|} 

\begin{document}

\title{Observer-independence in the presence of a horizon}

\author{Ian T. Durham$^{1}$}
\email[]{idurham@anselm.edu}
\affiliation{Department of Physics, Saint Anselm College, Manchester, NH 03102}
\date{\today}

\begin{abstract}
In the famous thought experiment known as Wigner's friend, Wigner assigns an entangled state to the composite quantum system consisting of his friend and her observed system. In the context of this thought experiment, Brukner recently derived a no-go theorem for observer-independent facts, i.e. those facts that would be common to both Wigner and his friend, and initial experimental tests have given it strong support. In this article, we demonstrate a loophole in the theorem that manifests in the presence of a horizon even if Wigner and his friend are on the same side of the horizon. The loophole requires that Wigner and his friend both unambiguously demonstrate that they are in inertial frames. We then argue that unambiguously proving that they are both in inertial frames requires unambiguously proving the existence of Unruh radiation. However, we also argue that if they cannot prove this, then there are still limitations on observer-independent facts that require that some of those facts remain unknowable. As such, observer-independent facts remain illusory.
\end{abstract}

\maketitle

\section{Introduction}
In 1961 Eugene Wigner introduced a curious thought experiment that has recently garnered renewed interest in the physics literature~\cite{Wigner:1961aa}. In the experiment, now known as ``Wigner's friend'', the aforementioned friend performs a measurement on a quantum system inside a sealed laboratory. A so-called ``super-observer'' (played by Wigner himself in the original paper) is placed outside the laboratory. The outcome of the friend's measurement is reflected in some property of the device that is performing the measurement, e.g. a pointer reading or audible device click. Wigner describes the process unitarily based solely on the information to which he has access. At the conclusion of the process the friend's description of the original quantum state consists of a projection of that state into one that corresponds to the outcome reported by the device. By contrast, Wigner assigns a specific entangled state to the system and the friend which he can verify through some further experiment (e.g. communicating with the friend). The main question that Wigner provoked with his thought experiment was what happens to the state from his (Wigner's) point-of-view when his friend observes a definite outcome? Does the state collapse for Wigner at that exact moment or does it only collapse for Wigner when he receives information about the result of his friend's measurement? If it is the latter, then how can we reconcile the two apparently different accounts of the original measurement process?

Wigner's original intent with his thought experiment was to support his view that consciousness was necessary in order to collapse the wave function. It nevertheless serves as an interesting way to compare various interpretations of quantum mechanics. In particular, in objective collapse theories, Wigner's state assignment can be statistically disproven by carrying out many verifications hence showing that the state is a definite, objective property of the universe~\cite{Ghirardi:1986aa,Diosi:1989aa,Penrose:1996aa}. This is in contrast to other interpretations where the state is only a relative property, e.g. it is projected relative to the friend and in a superposition relative to Wigner~\cite{Rovelli:1996aa,Fuchs:2017aa}. Both Wigner's and his friend's description of the state are equally valid in a relative sense since they accurately describe the world Wigner and his friend experience respectively. When they finally communicate, Wigner can update his description based on information he receives from his friend concerning the results of the friend's measurement. Thus, strictly speaking, there is no inconsistency with quantum theory in this case. In other words, Wigner's friend poses no apparent problems for epistemic interpretations of quantum theory.

There has been renewed interest in this thought experiment in recent years, most notably in works by Brukner~\cite{Brukner:2015aa,Brukner:2018aa} and by Frauchiger and Renner~\cite{Frauchiger:2018aa}. In particular, Brukner derived a Bell-type no-go theorem for observer-independent facts that aims to show that there can be no theory in which Wigner's and Wigner's friend's facts can jointly be considered to be locally objective properties of the universe. That is, Brukner's theorem denies that either Wigner's or his friend's description of the state is a (locally) objective property of a given theory. As Brukner makes clear, the objective properties or ``facts'' described are understood to mean ``immediate experiences of observers.'' That is, it may refer to what some interpretations of quantum mechanics deems to be ``real'' (e.g. wave functions, Bohmian trajectories, etc.) only to the extent that these directly give rise to some sort of observable fact such as a pointer reading or detector click. Brukner's theorem was recently experimentally verified to more than five standard deviations by Proietti, et. al.~\cite{Proietti:2019aa}.

The relative state description of Wigner and his friend is analogous to the situation described by another famous thought experiment: Einstein's twin paradox, the first explanation of which was due to Paul Langevin~\cite{Langevin:1911aa}. In the twin paradox a pair of synchronized clocks becomes un-synchronized when one clock undergoes relativistic acceleration while the other remains inertial. In Langevin's retelling, observers co-moving with the accelerated clock remain in regular contact with observers who remain in the frame of the inertial clock. Thus it is that the inertial observers will appear to see the accelerated clock slow down while observers co-moving with that same clock will see it behave normally. Each description is equally valid in a relative sense since both clocks (and their associated observers) maintain separate worldlines. When the worldlines re-intersect, both sets of observers can simply compare their respective results in order to arrive at a consistent description of the state of the clock at that instant. As such, special relativity contains no inconsistency in this case.

In this article we investigate a scenario in which the experimental setup employed by Brukner in the proof of his theorem is embedded in a twin-paradox-like setting. The relative acceleration that is introduced occurs between two identical laboratories, each of which contains an observer in the role of Wigner's friend who makes a measurement on a quantum system in their possession. Likewise, each laboratory includes a co-moving external super-observer in the role of Wigner who communicates with the observer inside the laboratory and whose description of the quantum system depends on the information received. We then extend the discussion to a single Wigner's friend-type experiment in which the relative acceleration is then between Wigner and his friend.

We show that this setup brings to light a loophole in the no-go theorem in the presence of a horizon. However, we also show that if this loophole proves fatal for the theorem and restores observer-independent facts, it comes at the price of fully knowing those facts, suggesting that the existence of such facts remains illusory. We begin by briefly reviewing Brukner's theorem and we reformulate his proof in terms of free bosonic fields. We then show the dependence of this proof the acceleration one or both halves of the experiment and argue that this amounts to a coordinate transformation. We further argue that proofs of theorems should be coordinate-independent. However, we also show that proving this would require proving the existence of Unruh radiation. If this proves to not be possible (a fact that is somewhat open to debate~\cite{Kiaka:2018aa}), then we show that there must still be some limitations on these facts. Specifically we argue that some of them must remain unknowable. We begin with a review of Brukner's theorem and its proof, but rather than using an entangled spin-1/2 system, our model uses an entangled two-mode free bosonic field which makes analysis in the presence of acceleration simpler since it reduces a portion of the problem to a discussion of the Unruh effect.

\section{Observer-independent facts for inertial observers}\label{sec2}

A typical Wigner's friend thought experiment involves some two-level quantum system that can give rise to two outcomes upon measurement. The outcomes are recorded by a measurement apparatus that eventually is read by Wigner's friend and committed to memory. The measurement apparatus as well as Wigner's friend are taken to be inside an isolated laboratory. Wigner is placed outside this laboratory and can perform a quantum measurement on the overall system consisting of the two-level quantum system and the laboratory. It is usually assumed that all experiments are carried out a sufficient number of times in order to ensure reliable statistics.

Deutsch proposed a model of the Wigner's friend thought experiment for which it was possible for Wigner to gain direct knowledge about whether or not the friend had observed a definite outcome without her (the friend) having to reveal which outcome she observed~\cite{Deutsch:1985aa}. For example, let us assume that a detector that is sensitive to two modes $j$ and $k$ of a free bosonic field in Minkowski space $\mathcal{M}$ is used to make a measurement on such a field that is prepared in the state $\ket{\phi^+}_B=\frac{1}{\sqrt{2}}(\ket{0_j}_{B}^{\mathcal{M}} + \ket{1_k}_{B}^{\mathcal{M}})$ where $\ket{0}_{B}^{\mathcal{M}}$ and $\ket{1}_{B}^{\mathcal{M}}$ refer to the vacuum and single-particle states respectively. After the measurement has been completed, the measurement apparatus is found to be in one of many perceptively different macroscopic configurations corresponding to, for example, a particular pointer reading or an audible click. In Deutsch's model, no assumptions need to be made concerning the friend's formal description of the result. It is simply enough that she perceives a definite outcome.

Wigner then uses quantum theory to describe the friend's measurement as a unitary transformation. The possible states of the field $\ket{0_j}_{B}^{\mathcal{M}}$ and $\ket{1_k}_{B}^{\mathcal{M}}$ are assumed to be entangled with the perceptively different macroscopic configurations of the apparatus, the laboratory, and the friend's memory. We can represent these configurations using a pair of orthogonal states $\ket{F_0}_F$ and $\ket{F_1}_F$ respectively. The state of the composite system consisting of the field mode and the friend's laboratory is thus
\begin{equation}
\ket{\Phi}_{BF} = \frac{1}{\sqrt{2}}\left(\ket{0_j}_{B}^{\mathcal{M}}\ket{F_0}_F + \ket{1_k}_{B}^{\mathcal{M}}\ket{F_1}_{F}\right).
\label{BFstate}
\end{equation}
The particular phase between the two amplitudes of equation~\eqref{BFstate} is specified by Wigner via the measurement interaction. It is his specification of this phase that avoids the necessity of describing the state as an incoherent mixture of the two possibilities. Wigner can verify his state assignment by performing a Bell state measurement in the bases
\begin{align}
\ket{\Phi^{\pm}}_{BF} & =\frac{1}{\sqrt{2}}(\ket{0_j}_{B}^{\mathcal{M}}\ket{F_0}_F \pm \ket{1_k}_{B}^{\mathcal{M}}\ket{F_1}_{F}), \nonumber \\
\ket{\Psi^{\pm}}_{BF} & =\frac{1}{\sqrt{2}}(\ket{0_j}_{B}^{\mathcal{M}}\ket{F_1}_F \pm \ket{1_k}_{B}^{\mathcal{M}}\ket{F_0}_{F}).
\label{bellbases}
\end{align} 
In Deutsch's proposal, Wigner obtains direct knowledge about whether or not the friend actually observed a definite outcome without requiring that the friend reveal the result. For instance, the friend could pass a note through a slot in the laboratory door on which is written either ``I have observed a definite outcome'' or ``I have \textit{not} observed a definite outcome'' as shown in Figure~\ref{dfig}.
\begin{figure}
\includegraphics[width=0.485\textwidth]{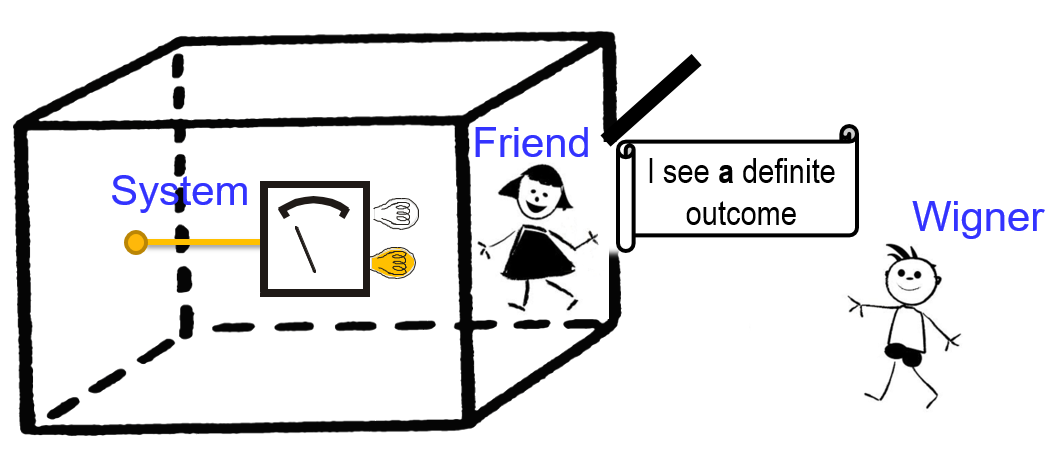}
\caption{\label{dfig} Wigner can obtain information regarding whether or not his friend has performed the measurement without learning the result of the measurement. For instance, the friend could simply pass a note through the laboratory door that indicates whether or not the friend obtained a definite outcome. (Figure courtesy of \u{C}aslav Brukner.)}
\end{figure} 
As long as the message does not contain any information regarding the actual outcome of the friend's measurement, the state of the encoded message can be factored out of the total state which then would read
\begin{align}
\ket{\Phi}_{BFI} & = \frac{1}{\sqrt{2}}\left(\ket{0_j}_{B}^{\mathcal{M}}\ket{F_0}_F + \ket{1_k}_{B}^{\mathcal{M}}\ket{F_1}_{F}\right) \nonumber \\
& \quad \otimes\ket{\textrm{``I have observed a definite outcome.''}}_I.
\label{BFIstate}
\end{align}
(Here we are assuming that the friend always observes a definite outcome when performing a measurement.) Because the state of the message can be factored out, we see that Wigner can obtain direct evidence for the existence of his friend's facts without knowing what those facts are. As such, both Wigner's facts and his friend's facts appear to coexist.

In a framework in which we can account for observer-independent facts, we should be able to jointly assign truth values to the observational statements $A_1$: ``Wigner's friend's measurement apparatus indicates the bosonic field is in the vacuum state'' and $A_2$: ``Wigner's measurement apparatus indicates the overall state is $\Phi$.'' Wigner, of course, can learn the truth value of either of these two statements. If he performs a Bell measurement, he obtains the truth value for $A_2$ whereas if he, for example, simply opens the laboratory door and speaks to his friend, he can obtain the truth value for $A_1$. But in order for observer-independent facts to exist, we must be able to assign a truth value to both $A_1$ and $A_2$ \textit{independently} of which measurement Wigner performs, i.e. independently of whether he makes a Bell measurement or opens the lab door. In other words, if the outcome $\Phi$ is observer-independent, then $A_1$ is true regardless of whether or not Wigner actually makes that measurement.

Brukner formalizes this by postulating that the truth values of the propositions $A_i$ for all observers form a Boolean algebra that is equipped with a countably additive positive measure $p(A)\ge 0$ for all propositions that correspond to the probability that a given proposition is true. For the scenario described here, we assume that we can jointly assign truth values ($+1=$ true, $-1=$ false) to the statements $A_1$ and $A_2$ and thus may also assign joint a joint probability $p(A_1 = \pm 1,A_2 = \pm 1)$. Brukner's no-go theorem, which is a Bell-type theorem, then uses the fact that $A_1$ and $A_2$ do not commute and is stated as follows~\cite{Brukner:2018aa}:

\begin{theorem}[Brukner's no-go theorem]\label{bruktheo}
The following statements are incompatible, i.e. they lead to a contradiction:
\begin{enumerate}
\item
Quantum predictions hold at any scale, even if the measured system contains objects as large as an ``observer'' (including her laboratory, memory, etc.). This is the assumption that quantum theory is universally valid.
\item
The choice of the measurement settings of one observer has no influence on the outcomes of any other other distant observer(s). This is the assumption of locality.
\item
The choice of measurement settings is statistically independent from the rest of the experiment. This is the freedom of choice assumption.
\item
One can jointly assign truth values to the propositions about observed outcomes (``facts'') of different observers (as just described).
\end{enumerate}
\end{theorem}

\begin{proof}
Consider a pair of super-observers (Alice and Bob) who play the role of Wigner in a Deutsch-like Wigner's friend experiment as shown in Figure~\ref{bellfig}. 
\begin{figure}
\includegraphics[width=0.485\textwidth]{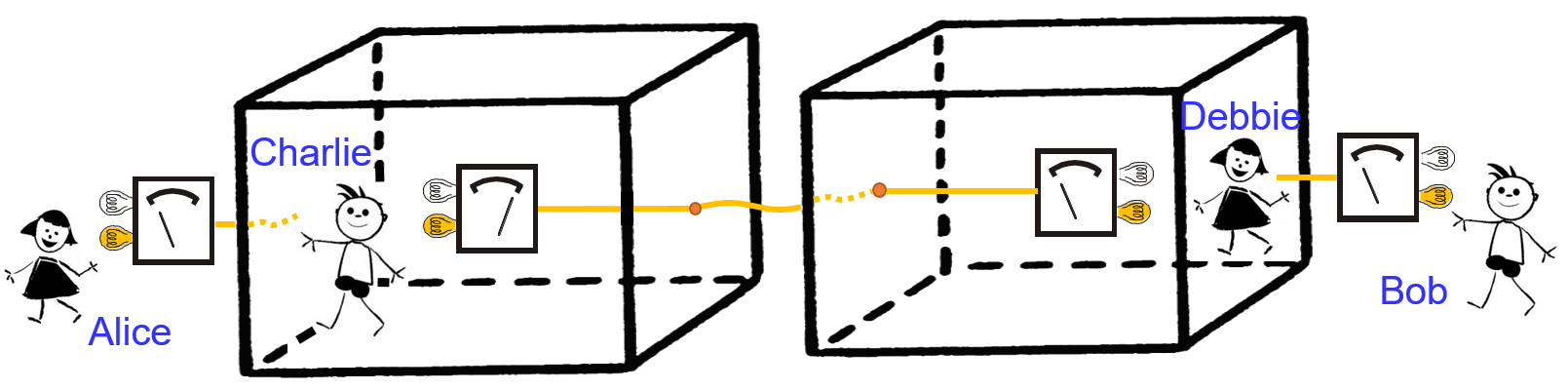}
\caption{\label{bellfig} A Bell-type experiment with two entangled Wigner's friend scenarios. Alice and Bob are the super-observers who perform measurements on the respective labs containing Charlie and Debbie. Charlie and Debbie each perform a quantum measurement on one-half of an entangled system. (Figure courtesy of \u{C}aslav Brukner.)}
\end{figure}
That is, these super-observers may each carry out experiments on a system that includes a laboratory containing an observer (Charlie and Debbie respectively) who performs an experiment on a free bosonic field mode. We can thus perform a Bell-type test for which Alice chooses between measurements $A_1$ and $A_2$ which correspond to the statements that Charlie and Alice can respectively make about their measurement outcomes. Likewise, $B_1$ and $B_2$ are similar measurements that apply to Debbie and Bob respectively. The assumptions (2), (3), and (4) define hidden variables that pre-determine the values of $A_1$, $A_2$, $B_1$, and $B_2$ to be either $+1$ or $-1$. As such we can also define a joint probability $p(A_1,A_2,B_1,B_2)$ whose marginals satisfy the Clauser-Horne-Shimony-Holt (CHSH) inequality: $S = \avg{A_1B_1} + \avg{A_1B_2} + \avg{A_2B_1} - \avg{A_2B_2} \le 2\sqrt{2}$.

Suppose that Charlie and Debbie initially share an entangled two-mode free bosonic field that is in the state
\begin{equation}
\ket{\psi}_{B_1B_2}=-\sin\frac{\theta}{2}\ket{\phi^-}_{B_1B_2}+\cos\frac{\theta}{2}\ket{\psi^+}_{B_1B_2}
\label{initstate}
\end{equation}
where
\begin{align}
\ket{\phi^-}_{B_1B_2} & =\frac{1}{\sqrt{2}}\left(\ket{0_j}_{B_1}^{\mathcal{M}}\ket{0_k}_{B_2}^{\mathcal{M}} - \ket{1_j}_{B_1}^{\mathcal{M}}\ket{1_k}_{B_2}^{\mathcal{M}}\right) \\
\ket{\psi^+}_{B_1B_2} & =\frac{1}{\sqrt{2}}\left(\ket{0_j}_{B_1}^{\mathcal{M}}\ket{1_k}_{B_2}^{\mathcal{M}} + \ket{1_j}_{B_1}^{\mathcal{M}}\ket{0_k}_{B_2}^{\mathcal{M}}\right)
\end{align}
and, again, $j$ and $k$ refer to the respective modes. Here Charlie controls mode $j$ and Debbie controls mode $k$ where each is inside their own laboratory. The state given by equation~\eqref{initstate} can be obtained by applying the appropriate pseudospin operators to the singlet state~\cite{Chen:2001aa}.

For Alice and Bob, the initial state together with the overall state of the laboratories is
\begin{equation}
\ket{\Psi_0} = \ket{\psi}_{B_1B_2}\ket{0}_C\ket{0}_D
\end{equation}
where the states $\ket{0}_C$ and $\ket{0}_D$ require no further characterization except to say that the observers are capable of completing a measurement.

Now we assume that Charlie and Debbie each have access to a detector that is sensitive to the single mode that is under their control (e.g. $j$ for Charlie and $k$ for Debbie). They each perform a measurement of their respective modes which amounts to determining if their mode is in the vacuum or single-particle state. From the point of view of Alice and Bob, these measurements are described by unitary transformations. Once the measurements are complete, we assume that the overall state of the entire system becomes
\begin{equation}
\ket{\tilde{\Psi}} = -\sin\frac{\theta}{2}\ket{\Phi^+} + \cos\frac{\theta}{2}\ket{\Psi^-}
\end{equation}
where
\begin{align}
\ket{\Phi^+} & =\frac{1}{\sqrt{2}}\left(\ket{A_{g}}\ket{B_{g}} - \ket{A_{e}}\ket{B_{e}}\right), \\\ket{\Psi^-} & = \frac{1}{\sqrt{2}}\left(\ket{A_{g}}\ket{B_{e}}+\ket{A_{e}}\ket{B_{g}}\right)
\end{align}
and
\begin{align}
\ket{A_{g}}=\ket{0_s}_{B_1}^{\mathcal{M}}\ket{C_{0_s}}_C, & \;\; \ket{A_{e}}=\ket{1_s}_{B_1}^{\mathcal{M}}\ket{C_{1_s}}_C, \label{alicestate}\\
\ket{B_{g}}=\ket{0_k}_{B_2}^{\mathcal{M}}\ket{D_{0_k}}_D, & \;\; \ket{B_{e}}=\ket{1_k}_{B_2}^{\mathcal{M}}\ket{D_{1_k}}_D \label{bobstate}
\end{align}
where the subscript $g$ refers to the vacuum state and the subscript $e$ refers to the single-particle state. Here the states $\ket{C_{0_s}}$ and $\ket{C_{1_s}}$ correspond to orthogonal states of Charlie's lab associated with the two measurement outcomes (and similarly for Debbie's lab). The details of these lab states is unimportant. All that matters is that they are orthogonal.

We next define two sets of binary observables in analogy to the Pauli spin operators along the $z$ and $x$ axes respectively:
\begin{align}
A_z & = \ket{A_{g}}\bra{A_{g}} - \ket{A_{e}}\bra{A_{e}} \label{az} \\
A_x & = \ket{A_{g}}\bra{A_{e}} + \ket{A_{e}}\bra{A_{g}} \label{ax}.
\end{align}
Similar operators can be defined for $B_z$ and $B_x$. In the Bell experiment, Alice chooses between $A_1 = A_z$ and $A_2 = A_x$ and likewise for Bob. In this case, $A_1$ and $B_1$ represent a Wigner's friend type of measurement while $A_2$ and $B_2$ represent a Wigner type of measurement. If we set $\theta = \pi/4$, then we find a violation of a CHSH-type inequality of
\begin{equation}
S = \avg{A_1B_1} + \avg{A_1B_2} + \avg{A_2B_1} - \avg{A_2B_2}\le 2.
\label{chsh}
\end{equation}
with a value of $S = 2\sqrt{2}$ which is a maximal violation in agreement with Brukner's results for spin-1/2 systems.\end{proof}

\section{Observer-independent facts in the presence of a horizon}\label{horizon}
Let us now consider the situation if there is relative acceleration between the two halves of the experiment. That is, let us uniformly co-accelerate Bob, Debbie, and her laboratory with some proper acceleration $a$. It is worth noting that the following result holds even if both halves of the experiment are co-accelerated but, for simplicity we only accelerate one half~\cite{Fuentes-Schuller:2005aa}.

Proper acceleration is interpreted as a physical acceleration, i.e. an acceleration that is measured by some device such as an accelerometer. It is therefore measured relative to an inertial observer who is momentarily at rest with respect to the object under acceleration. In this case, the states of mode $k$ must be specified in Rindler coordinates in order to properly describe what Debbie and Bob observe. If we consider just a single spatial dimension $z$, the worldlines of uniformly accelerated observers in Minkowski space correspond to hyperbolae to the left (region I) and right (region II) of the origin on a spacetime diagram as in Figure~\ref{rindlerwedge}. 
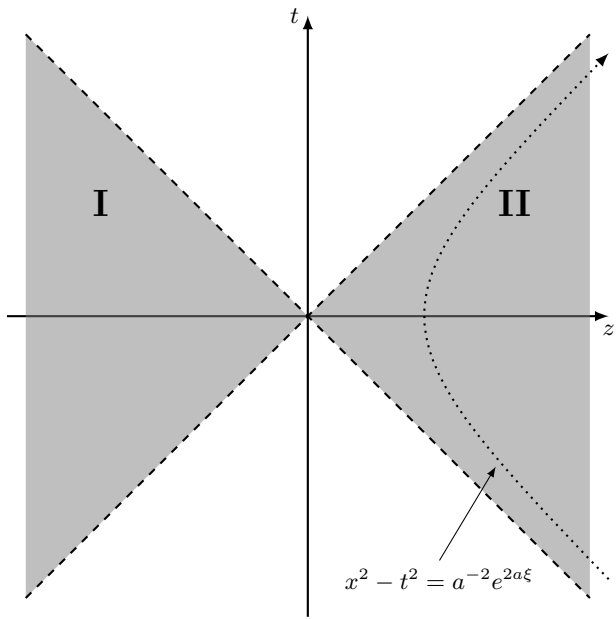
\begin{figure}
\begin{center}
\begin{tikzpicture}
\draw[thick,-latex] (-4,0) -- (4,0);
\node[below] at (4,0) {$z$};
\draw[thick,-latex] (0,-4) -- (0,4);
\node[left] at (0,4) {$t$};
\fill[gray,opacity=0.5] (-3.75,-3.75) -- (0,0) -- (-3.75,3.75) -- (-3.75,-3.75);
\fill[gray,opacity=0.5] (3.75,3.75) -- (0,0) -- (3.75,-3.75) -- (3.75,3.75);
\draw[thick,dashed] (-3.75,-3.75) -- (3.75,3.75);
\draw[thick,dashed] (-3.75,3.75) -- (3.75,-3.75);
\node at (-2.75,1.5) {\Large{\textbf{I}}};
\node at (2.75,1.5) {\Large{\textbf{II}}};
\draw[thick,dotted,-latex] (4,-3.5) .. controls (0.75,-0.125) and (0.75,0.125) .. (4,3.5);
\draw[-latex] (1.75,-3.25) -- (2.5,-2);
\node at (1.75,-3.5) {$x^2 - t^2 = a^{-2}e^{2a\xi}$};
\end{tikzpicture}
\end{center}
\caption{\label{rindlerwedge} An accelerated observer in Minkowski space is a static observer in a Rindler space. Worldlines of accelerated observers in Minkowski space form a parabola $x^2 -t^2 = a^{-2}e^{2a\xi}$ for constant $a$ and $\xi$. Observers on such a trajectory in region II are causally disconnected from region I and vice-versa.}
\end{figure}
These regions are bounded by light-like asymptotes that form the Rindler horizon. Rindler coordinates are defined as
\begin{align}
t = a^{-1}e^{a\xi}\sinh a\tau, & \;\; z = a^{-1}e^{a\xi}\cosh a\tau, \;\; |z|<t \nonumber \\
t = -a^{-1}e^{a\xi}\sinh a\tau, & \;\; z = a^{-1}e^{a\xi}\cosh a\tau, \;\; |z|>t
\label{rindler}
\end{align}
where $a$ is the proper acceleration as defined above, $\xi$ is a space-like coordinate, and $\tau$ is the proper time. The Minkowski vacuum state, 
\begin{equation}
\ket{0}^{\mathcal{M}}=\Pi_j \ket{0_j}^{\mathcal{M}},
\label{minkvac}
\end{equation}
which is defined as the absence of any particle excitation in any of the $j$ modes can be expressed in terms of a product of two-mode squeezed states of the Rindler vacuum~\cite{Walls:2008aa} as
\begin{equation}
\ket{0_k}^{\mathcal{M}} \sim \frac{1}{\cosh r}\sum_{n=0}^{\infty}\tanh^n r \ket{n_k}_I \ket{n_k}_{II}, 
\label{minktwomode}
\end{equation}
where $\cosh r = (1 - e^{-2\pi\Omega})^{-1/2}$, $\Omega := \abs{k}c/a$, and $a$ is the acceleration. The states $\ket{n_k}_I$ and $\ket{n_k}_{II}$ are Fock states and refer to the mode decomposition in regions I and II of Rindler space respectively. Likewise, the single-particle states can be written as~\cite{Fuentes-Schuller:2005aa}
\begin{equation}
\ket{1_k}^{\mathcal{M}} = \frac{1}{\cosh^2 r}\sum_{n=0}^{\infty}\tanh^n r\sqrt{n+1}\ket{(n+1)_k}_I\ket{n_k}_{II}.
\label{minkpart}
\end{equation}
For the sake of argument, let us assume that Debbie, her laboratory, and Bob are all uniformly accelerated in a direction such that they are causally disconnected from region I, i.e. they follow a trajectory given by the hyperbola $x^2 - t^2 = a^{-2}e^{2a\xi}$ as in Figure~\ref{rindlerwedge}. As such, equation~\eqref{bobstate} becomes
\begin{align}
\ket{B_g} & = \frac{1}{\cosh r}\sum_{n=0}^{\infty}\tanh^n r \ket{n_k}_I\ket{n_k}_{II}\ket{C_{0_k}}_C \nonumber \\
\ket{B_e} & = \left(\frac{1}{\cosh^2 r}\sum_{n=0}^{\infty}\tanh^n r\sqrt{n+1}\right) \nonumber \\
& \quad\quad \times \ket{(n+1)_k}_I\ket{n_k}_{II}\ket{C_{1_k}}_C. \label{rindlerbob}
\end{align}
Since Bob, Debbie, and her laboratory are not causally connected to region II, we must trace over the states in this region which gives a mixed state, 
\begin{equation}
\rho^{I}_{AB} = \textrm{tr}_{II}\ket{\tilde{\Psi}}\bra{\tilde{\Psi}}.
\label{mixedstate}\end{equation} 
The expansion of this density matrix is given in appendix~\ref{app}.

We can then form a pair of binary observables for Bob in analogy to equations~\eqref{az} and~\eqref{ax} and, once again setting $\theta=\pi /4$, we can test a CHSH-type inequality of the form shown in equation~\eqref{chsh}. The expectation values in this inequality are of the form $\avg{A_1B_1} = \textrm{tr}(\rho^{I}_{AB}A_1B_1)$. The problem with evaluating this inequality, however, is that equation~\eqref{rindlerbob} includes an infinite (vector) sum over all values of $n$ which represents the fact that Fock space is infinite in size. However, the sum is convergent and the size of the Fock space can be truncated by setting $n_{\textrm{max}}=N$ when $\tanh^Nr < \epsilon$ for some value of $\epsilon$. That is, we choose to ignore any vectors $\tanh^n r\ket{n}$ in the sum for which $\tanh^n r$ is sufficiently small. Nevertheless, the size of the density matrix is still quite large even for small values of $N$ (e.g. it is $81\times 81$ for $N=3$). As such it is necessary to evaluate this numerically. Of particular interest to us here is the dependence of $S$ on the proper acceleration $a$.

We chose to implement our model using Python's QuTiP package which is especially well-suited for working with Fock states~\cite{Johansson:2012aa,Johansson:2013aa} and our code is given in Appendix~\ref{code}. It turns out that the value of $\epsilon$ must be chosen carefully to ensure numerical stability while also not sacrificing any relevant physics. It is reasonable to expect that the dependence of $S$ on $a$ ought to be consistent with the dependence of the logarithmic negativity on $r$ as given in the work of Fuentes and Mann~\cite{Fuentes-Schuller:2005aa}.  Values of $\epsilon$ below roughly 0.1 demonstrate strong instabilities for values of $r$ ranging to approximately 6.0. With $\epsilon = 0.1$, however, we obtain results that are consistent with those of Fuentes and Mann for the logarithmic negativity. Our results are shown in Figure~\ref{sversusa}.
\begin{figure}
\includegraphics[width=0.485\textwidth]{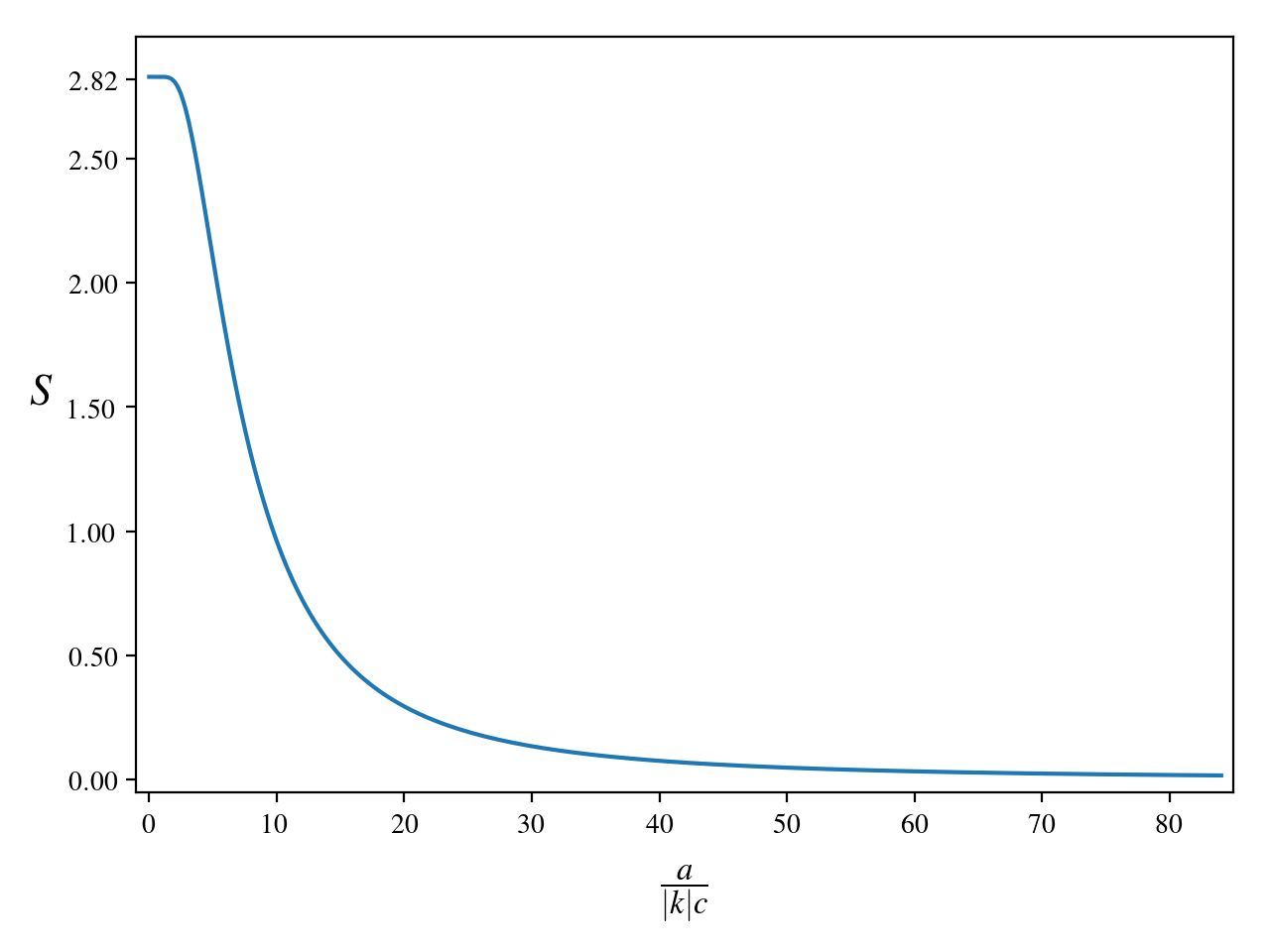}
\caption{\label{sversusa} A plot of the value of $S$ versus $a/\abs{k}c$. The transition to classicality occurs when $a/\abs{k}c\sim 5.3$. (Plot produced with Python's matplotlib package~\cite{Hunter:2007}.)}
\end{figure}
Classicality is reached when $a/\abs{k}c\sim 5.3$. The question is whether or not this necessarily proves fatal to the theorem. 

\section{Implications}\label{imply}
The most persuasive argument that can be made in favor of the theorem is that it only needs to be true once in order for its conclusions to be valid. Clearly a violation of the associated CHSH inequality is necessary to prove the theorem. That is, as long as the CHSH inequality presented in the situation here is violated under \textit{some} condition, then there is \textit{some} evidence that observer-independent facts are not universal. But is this also sufficient? 

A counter-argument runs thus. The representation of the Minkowski vacuum state (equation~\eqref{minkvac}) as a two-mode squeezed state by equation~\eqref{minktwomode} can be interpreted as a coordinate transformation between Minkowski space and Rindler space. The question is whether this is merely a coordinate transformation or if it represents some real, physical transformation. The answer to that question is not so subtle.

Suppose we chose to co-accelerate both halves of the experiment, i.e. there would be no \textit{relative} acceleration between the Alice/Charlie half and the Bob/Debbie half. Since accelerated observers in Minkowski space are static observers in Rindler space, the two halves of the experiment should appear approximately inertial \textit{relative to one another} in Rindler coordinates. Yet it is well-known that such observers will nevertheless measure each other as being in a thermal state. This fact was first discovered independently by Fulling, Davies, and Unruh~\cite{Fulling:1973aa,Davies:1975aa,Unruh:1976aa}. In fact the representation of the vacuum and single-particle states as equations~\eqref{minktwomode} and~\eqref{minkpart} was first shown by Davies and Unruh in their papers. 

The validity of Brukner's proof then rests on the ability of the two halves of the experiment to unambiguously demonstrate that they are in universally inertial frames. In other words, it is imperative that the coordinate transformation between Minkowski and Rindler coordinate systems represent a real, physical transformation and not merely a coordinate transformation. If it is merely a coordinate transformation, then Brukner's theorem would not hold because its proof would be coordinate-dependent. If, on the other hand, it is a real, physical transformation, then (at least theoretically) the situation represented in Brukner's proof would represent a different physical situation than that described by two co-accelerating pairs of observers. In that case, observer-independent facts would not exist in at least one, real physical situation and thus Brukner's theorem would have to be valid.

In some sense, then, unambiguously showing the validity of Brukner's theorem would require proving the existence of Unruh radiation. Of course it has long been assumed that Unruh radiation does exist and we simply do not yet have the ability to actually measure it due to the very high accelerations required. However, it was recently shown that \textit{massive} Unruh particles cannot be directly observed~\cite{Kiaka:2018aa}. If such a result were ever found for massless Unruh particles, Brukner's theorem could not be unambiguously proven.

Nevertheless, as we are about to show, even if Brukner's theorem could not be proven unambiguously, the existence of observer-independent facts would have certain limitations and the proof of their existence would remain illusory.

\section{Limitations on observer-independence}\label{implicate}
Brukner has suggested that one possible way around this problem is to assume that since Bob, Debbie, and her laboratory are beyond a horizon, that in some sense they don't exist from the standpoint of Alice and Charlie thus eliminating the loophole~\cite{Brukner:2019aa}. We now describe a scenario in which it is possible for Alice and Charlie to maintain knowledge of Bob's and Debbie's existence even if they pass beyond a horizon. The scenario relies on the fact that mass effects are not shielded by horizons. As such, we argue that the loophole in the theorem remains.

If Bob (and Debbie) exist then Alice should be able to assign a truth value to the proposition $A_3$: ``Bob, Debbie, and her laboratory exist.'' In order to address Brukner's point, it is merely necessary for Bob, Debbie, and her laboratory to possess some objective property that is, in some way, detectable by Alice regardless of which side of a horizon they are on. Since mass effects are not shielded by horizons, we can assume that Bob, Debbie, and her laboratory collectively possess a (relativistically invariant) mass $m$~\footnote{In order to ensure relativistic invariance, we define the mass here as the magnitude of the four-momentum.}. Their possession of property $m$ ensures that Alice can assign a truth value to proposition $A_3$ as long as $m$ is, in some way, distinctly measurable by Alice in that Alice can associate a value for $m$ to the system of Bob, Debbie, and her laboratory.

Let us suppose that Bob's half of the experiment is near an object of mass $M$. Further suppose that $m$ is a non-trivial fraction of $M$. The requirement of non-triviality ensures the distinctness of $m$ for the system of Bob, Debbie, and her laboratory. We assume that Alice has the ability to monitor the masses $m$ and $M$ and we further assume that she can initially distinguish between $m$ and $M$ (perhaps through some astrophysical method). We initially require this distinguishability to ensure that Alice can assign a definite truth value to proposition $A_3$ since, if she can't tell if $m$ is independent of $M$, she can't know whether Bob's half of the system actually exists as an entity unto itself.

Brukner's argument is that if at any point Bob's half of the experiment were to cross a horizon then Alice would no longer be able to say that Bob's half of the experiment exists since the two halves would be causally disconnected from one another. Now suppose that Bob's half of the experiment has crossed some sort of horizon, in the direction of which (from Alice's perspective) also lies $M$. For example, Bob's half of the experiment could cross the event horizon of a black hole of mass $M$. From Alice's perspective the total mass in the direction of the horizon remains $m+M$ since horizons do not shield mass effects. While $m$ may no longer be distinguishable from $M$ from Alice's point-of-view (for instance, if it crossed the event horizon of $M$), nevertheless $m$ still has a measurable effect on the surroundings. For instance, if $M$ is a black hole and the horizon is an event horizon, when $m$ crosses this horizon the black hole's mass and event horizon radius increase. Though the behavior of the system being monitored by Alice has changed, it is still consistent with a total mass of $m+M$. In other words, though $m$ has crossed some sort of horizon, its mass effects persist and thus Alice is still able to assign a truth value to proposition $A_3$. As such, the failure of the proof of Brukner's theorem in the presence of a horizon is a real, measurable effect.

This has another consequence, however. Imagine, instead, the usual Wigner's friend experiment in the form proposed by Deutsch. Recall that In a framework in which we can account for observer-independent facts, we should be able to jointly assign truth values to the propositions $A_1$ and $A_2$ as defined in Section~\ref{sec2}. Let us add to this the proposition $A_4$: ``Wigner's friend exists.'' Note that this is different from proposition $A_3$ which applies to both a Wigner-type setting (Bob) and a Wigner's friend-type setting (Debbie and her laboratory). The existence of a truth value for $A_4$ is actually \textit{not} independent from the other truth values. In particular, a truth value of $A_4$ that is consistent with the existence of Wigner's friend, necessarily implies that truth values \textit{can} (but do not have to) exist for $A_1$ and $A_2$.

But now consider a situation in which Wigner's friend and her laboratory have a measurable mass $m$ and pass a horizon that, from Alice's point-of-view, is in the direction of some mass $M$. Now suppose that at some point prior to Wigner carrying out a measurement of type $A_1$ or type $A_2$, his friend and her laboratory pass the horizon. Due to the results outlined above, the friend and her laboratory still exist since they have a measurable mass which is not shielded by a horizon. As such, Wigner can still assign a truth value to proposition $A_4$ yet cannot assign a truth value to either proposition $A_1$ or $A_2$. As such it appears that \textit{it is possible for Wigner to know that truth values for $A_1$ and $A_2$ exist and yet be unable to determine them.}

This appears to offer an interesting compromise on the part of physics in regard to observer-independent facts. If Brukner's theorem were to fail, suggesting that observer-independent facts \textit{might} still exist, it also appears to be the case that this failure would come at the cost of fully knowing those facts. In other words, it may be that in order for observer-independent facts to exist, some of those facts must remain unknowable. It is as if Wigner is prevented by the laws of physics from taking any further actions in Figure~\ref{dfig} after first receiving the friend's message. Either way, observer-independence appears to be illusory.

\section{Conclusion}
In this article we have examined the recent no-go theorem for observer-independent facts proposed by Brukner and have shown that the proof of the theorem depends on the existence of Unruh radiation. Our model employed entangled modes of a free bosonic field and our results for the dependence of the Bell parameter $S$ as a function of relative acceleration $a$ are consistent with known results for the logarithmic negativity, though we note that similar results are derivable for fermionic systems~\cite{Chaves-Otaola:2012aa}. However, we have also shown that if Brukner's theorem cannot be unambiguously proven, there remain limitations on observer-independent facts that prevent full knowledge of those facts.

It is worth noting here that the use of definite particle states for testing a situation such as this can prove problematic~\cite{Sorkin:1993aa,Dowker:2011aa}. Specifically, the bosonic field modes, though confined to the respective laboratories of Charlie and Debbie, are still highly non-local, which can lead to superluminal signaling between two observers \textit{within} the laboratory, one to the past of the spacelike hypersurface on which the observable in question is measured, and one to the future. While it is not clear that this applies in the scenario discussed in this paper, similar results can be obtained for fermionic systems which demonstrate similar behavior but do not suffer from the same non-locality issues~\cite{Chaves-Otaola:2012aa}. On the other hand, it has also been shown that massive Unruh particles cannot be directly measured~\cite{Kiaka:2018aa} and thus testing this would require the ability to work with massless fermionic fields.

\appendix
\section{Mixed state density matrix}\label{app}
The density matrix for the mixed state in the case in which Bob, Debbie, and her lab are accelerated is
\begin{widetext}
\begin{align}
\rho^{I}_{AB} & = \frac{\sin^2\left(\theta/2\right)}{2}\ket{A_g}^{\mathcal{M}}\ket{B_g}_{I}\bra{A_g}^{\mathcal{M}}\bra{B_g}_{I} - \frac{\sin^2\left(\theta/2\right)}{2}\ket{A_g}^{\mathcal{M}}\ket{B_g}_{I}\bra{A_e}^{\mathcal{M}}\bra{B_e}_{I} \nonumber \\
& -\frac{\sin\left(\theta/2\right)\cos\left(\theta/2\right)}{2}\ket{A_g}^{\mathcal{M}}\ket{B_g}_{I}\bra{A_g}^{\mathcal{M}}\bra{B_e}_{I} -\frac{\sin\left(\theta/2\right)\cos\left(\theta/2\right)}{2}\ket{A_g}^{\mathcal{M}}\ket{B_g}_{I}\bra{A_e}^{\mathcal{M}}\bra{B_g}_{I} \nonumber \\
& -\frac{\sin^2\left(\theta/2\right)}{2}\ket{A_e}^{\mathcal{M}}\ket{B_e}_{I}\bra{A_g}^{\mathcal{M}}\bra{B_g}_{I} + \frac{\sin^2\left(\theta/2\right)}{2}\ket{A_e}^{\mathcal{M}}\ket{B_e}_{I}\bra{A_e}^{\mathcal{M}}\bra{B_e}_{I} \nonumber \\
& +\frac{\sin\left(\theta/2\right)\cos\left(\theta/2\right)}{2}\ket{A_e}^{\mathcal{M}}\ket{B_e}_{I}\bra{A_g}^{\mathcal{M}}\bra{B_e}_{I} + -\frac{\sin\left(\theta/2\right)\cos\left(\theta/2\right)}{2}\ket{A_e}^{\mathcal{M}}\ket{B_e}_{I}\bra{A_e}^{\mathcal{M}}\bra{B_g}_{I} \nonumber \\
& -\frac{\sin\left(\theta/2\right)\cos\left(\theta/2\right)}{2}\ket{A_g}^{\mathcal{M}}\ket{B_e}_{I}\bra{A_g}^{\mathcal{M}}\bra{B_g}_{I} + \frac{\sin\left(\theta/2\right)\cos\left(\theta/2\right)}{2}\ket{A_g}^{\mathcal{M}}\ket{B_e}_{I}\bra{A_e}^{\mathcal{M}}\bra{B_e}_{I} \nonumber \\
& +\frac{\cos^2\left(\theta/2\right)}{2}\ket{A_g}^{\mathcal{M}}\ket{B_e}_{I}\bra{A_g}^{\mathcal{M}}\bra{B_e}_{I} + \frac{\cos^2\left(\theta/2\right)}{2}\ket{A_g}^{\mathcal{M}}\ket{B_e}_{I}\bra{A_e}^{\mathcal{M}}\bra{B_g}_{I} \nonumber \\
& -\frac{\sin\left(\theta/2\right)\cos\left(\theta/2\right)}{2}\ket{A_e}^{\mathcal{M}}\ket{B_g}_{I}\bra{A_g}^{\mathcal{M}}\bra{B_g}_{I} + \frac{\sin\left(\theta/2\right)\cos\left(\theta/2\right)}{2}\ket{A_e}^{\mathcal{M}}\ket{B_g}_{I}\bra{A_e}^{\mathcal{M}}\bra{B_e}_{I} \nonumber \\
& +\frac{\cos^2\left(\theta/2\right)}{2}\ket{A_e}^{\mathcal{M}}\ket{B_g}_{I}\bra{A_g}^{\mathcal{M}}\bra{B_e}_{I} +\frac{\cos^2\left(\theta/2\right)}{2}\ket{A_e}^{\mathcal{M}}\ket{B_g}_{I}\bra{A_e}^{\mathcal{M}}\bra{B_g}_{I}
\end{align}
\end{widetext}
where $\ket{A}^{\mathcal{M}}\ket{B}_I\bra{A}^{\mathcal{M}}\bra{B}_I = \textrm{tr}_{II}\left(\ket{A}\ket{B}\bra{A}\bra{B}\right)$ with $\ket{A}$ given by equation~\eqref{alicestate} and $\ket{B}$ given by equation~\eqref{rindlerbob}.

\section{Code for numerical simulation of acceleration dependence}\label{code}
The following Python code was used to generate Figure~\ref{sversusa}.
\begin{lstlisting}
# Brukner's no-go theorem under acceleration
# Author: Ian T. Durham

# Import packages

from math import tanh, cosh, sqrt, sin, cos, pi, log
from qutip import *
from numpy import arange
import matplotlib.pyplot as plt
from matplotlib import rc
rc('font',**{'family':'serif','serif':['Times']})
rc('text',usetex=True)

# Functions to create Minkowski vacuum and single-particle states

def vacuum(N,r):
    v = fock(N,0)
    for n in range(1,N-1):
        v += (tanh(r)**n)*fock(N,n)
    return v*(1/cosh(r))

def particle(N,r):
    p = fock(N,0)
    for n in range(1,N-1):
        p += (tanh(r)**n)*(sqrt(n+1))*fock(N,n+1)
    return p*(1/(cosh(r)**2))

# Create empty arrays for plotting

S = []
A = []

# Size of truncated Fock space

N = 3

for r in arange(0.0,2.0,0.01):

    # Minkowski states

    zero_m = fock(N,0)
    one_m = fock(N,1)

    # Rindler states

    zero_r = vacuum(N,r)
    one_r = particle(N,r)

    # Macroscopic lab states

    C_z = fock(N,0)
    C_o = fock(N,1)

    D_z = fock(N,0)
    D_o = fock(N,1)

    # Possible outcome states

    Ag = tensor(zero_m,C_z)
    Ae = tensor(one_m,C_o)

    Bg = tensor(zero_r,D_z)
    Be = tensor(one_r,D_o)

    # Reduced density matrix for the state

    Phi_p = (1/sqrt(2)) * (tensor(Ag,Bg) - tensor(Ae,Be))
    Psi_m = (1/sqrt(2)) * (tensor(Ag,Be) + tensor(Ae,Bg))

    Psi_t = -sin((pi/4)/2) * Phi_p + cos((pi/4)/2) * Psi_m

    rho = Psi_t*Psi_t.dag()

    # Alice and Bob measurement operators

    A1 = Ag*Ag.dag() - Ae*Ae.dag()
    A2 = Ag*Ae.dag() + Ae*Ag.dag()

    B1 = Bg*Bg.dag() - Be*Be.dag()
    B2 = Bg*Be.dag() + Be*Bg.dag()

    # Expectation values

    E1 = expect(rho,tensor(A1,B1))
    E2 = expect(rho,tensor(A1,B2))
    E3 = expect(rho,tensor(A2,B1))
    E4 = expect(rho,tensor(A2,B2))

    # Report value for S

    S.append(abs(E1+E2+E3-E4))
    if r == 0.0:
        A.append(0.0)
    else:
        A.append(-2*pi/(log(tanh(r)**2)))

# Plot

fig, ax = plt.subplots()
lines = ax.plot(A, S)
ax.set_yticks(list(ax.get_yticks()) + [2.82])
ax.set_ylim(-0.05,2.99)
ax.set_xlim(-1,85)

plt.xlabel(r'$\frac{a}{|k|c}$',fontsize=16,labelpad=10)
h = plt.ylabel(r'$S$',fontsize=16,labelpad=10)
h.set_rotation(0)
plt.tight_layout()
plt.show()
\end{lstlisting}

\begin{acknowledgements}
We thank \u{C}aslav Brukner for helpful discussions and for the use of his original figures. We also thank Alexander Smith for a detailed review of the manuscript. Additionally we acknowledge helpful discussions with Flavio del Santo, Martin Ringbauer, Maaneli Derakhshani and Howard Wiseman.
\end{acknowledgements}

\bibliographystyle{RS}

\bibliography{RSPADurham.bib}

\end{document}